\documentclass[12pt]{amsart}

\usepackage{amssymb,amsthm,amsmath}
\usepackage[numbers,sort&compress]{natbib}
\usepackage{color}
\usepackage{graphicx}
\usepackage{tikz}

\title[ Distributions of   resonances of   quasi-periodic operators]{  Distributions of   resonances of  supercritical  quasi-periodic operators }

\author{Wencai Liu}
\address[W. Liu]{ Department of Mathematics, Texas A\&M University, College Station, TX 77843-3368, USA} \email{liuwencai1226@gmail.com; wencail@tamu.edu}



\newcommand{\R}{\mathbb{R}}

\newcommand{\T}{\mathbb{T}}
\newcommand{\Z}{\mathbb{Z}}

\newcommand{\Q}{\mathbb{Q}}

\theoremstyle{plain}
\newtheorem{theorem}{Theorem}[section]
\newtheorem{corollary}[theorem]{Corollary}
\newtheorem{lemma}[theorem]{Lemma}
\newtheorem{proposition}[theorem]{Proposition}

\theoremstyle{definition}

\newtheorem{remark}[theorem]{Remark}

\begin{document}

	
	\begin{abstract}

We discover that 
		the  distribution of (frequency and phase) resonances   plays a role in determining  the spectral type of  supercritical  quasi-periodic Schr\"odinger operators.  In particular, we  disprove   the  second spectral transition line conjecture of 	 Jitomirskaya in the early  1990s.
		
	\end{abstract}
	\maketitle
	
	\section{Introduction}
In this paper, we  study  the quasi-periodic Schr\"odinger operator 
$H = H_{\lambda v,\alpha,\theta}$ defined on  $\ell^2(\Z)$: 
\begin{equation}\label{op1}
(H_{\lambda v ,\alpha,\theta}u)(n)=u({n+1})+u({n-1})+ \lambda v (\theta+n\alpha)u(n),  n\in\Z 
\end{equation}
where   $\lambda $ is the coupling constant, $v:\T=\R/\Z\to \R$ is the (analytic) potential,  $\alpha \in \R  \backslash \Q $ is the frequency, and $\theta \in \R$ is the phase. 
Letting  $v=2\cos 2\pi \theta$ in  \eqref{op1}, we obtain  the almost Mathieu operator 
\begin{equation}\label{op}
(H_{\lambda ,\alpha,\theta}u)(n)=u({n+1})+u({n-1})+ 2\lambda \cos2\pi  (\theta+n\alpha)u(n).   
\end{equation}
The almost Mathieu operator  is a popular model in both mathematics and physics.
We refer readers to  ~\cite{jmarx,yicm,jlz20,jcdm,sch1,congr}  for  the history and recent development of quasi-periodic Schr\"odinger operators.

%

By Avila's global theory ~\cite{agt},  the analytic quasi-periodic Schr\"odinger operator \eqref{op1} can be categorized into three cases  based on the behavior of the complexified  Lyapunov exponent of corresponding Schr\"odinger cocycles: subcritical, critical and supercritical.
For the almost Mathieu operator, the Lyapunov exponent  on the spectrum  only depends on $\lambda$ ~\cite{bj02}, so the three regimes are  characterized by the coupling constant $\lambda$, 
\begin{itemize}
	\item[1.]  subcritical: $|\lambda|<1$;
	\item[2.]  critical: $|\lambda|=1$;
	\item[3.]  supercritical: $|\lambda|>1$.
\end{itemize}
We remark that 
the three regimes of the almost Mathieu operator  are not surprising, which has already been seen from   the prediction of 
Aubry and Andr{\'e} ~\cite{aa}:
\begin{itemize}
	\item[1.] if $|\lambda|<1$, $H_{\lambda,\alpha,\theta}$  has purely absolutely continuous spectrum for all $\alpha\in\R$ and $\theta\in\R$;
	
	\item[2.] if $|\lambda|=1$, $H_{\lambda,\alpha,\theta}$  has purely  singular  continuous spectrum for all $\alpha\in\R\backslash\Q$ and $\theta\in\R$. 
	\item[3.] if $|\lambda|>1$, $H_{\lambda,\alpha,\theta}$  has only pure  point  spectrum for all $\alpha\in\R\backslash\Q$ and $\theta\in\R$. 
\end{itemize}

For the subcritical  (namely $|\lambda|<1$)   almost Mathieu operator, after many earlier  works ~\cite{j,ad08,aj10,gjls}, Avila finally  proved that the spectrum is purely absolutely continuous ~\cite{avi08} for any $(\alpha,\theta)\in\R^2$.   
From  the  almost reducibility   conjecture \cite{avila2010almost,av1}, 
subcritical  analytic quasi-periodic Schr\"odinger operators (for small $\lambda$,  the operator \eqref{op1} is always subcritical) 
always have purely  absolutely continuous spectrum.

It is expected that the critical quasi-periodic Schr\"odinger operator only supports the singular continuous spectrum. This has been recently proved for the  almost Mathieu operator \cite{j19,ajm}. 
We should mention that    critical analytic quasi-periodic Schr\"odinger operators do not typically happen \cite{agt}.

The supercritical  analytic quasi-periodic operator (for large $\lambda$,  the operator \eqref{op1} is always supercritical)  has  Anderson localization (only pure point spectrum with exponentially decaying eigenfunctions) for almost every $(\alpha,\theta)$ in $\R^2$ \cite{bg02,j,bgreen05}, but not  for all $\alpha\in\R\backslash\Q$  and $\theta\in\R$ since the arithmetics of frequencies $\alpha$ and phases $\theta$ play roles ~\cite{js94,gordon,as82,sar}.  This in particular implies that 
Part 3 of  Aubry-Andre's prediction does not hold  all $\alpha\in\R\backslash\Q$ and $\theta\in\R$.

From the discussion above,  one can see that  it is particularly  interesting to study supercritical quasi-periodic operators, which is exactly the main focus of this paper. 

By Kotani theory,   the  supercritical  analytic quasi-periodic operator   does not have  any absolutely continuous spectrum
~\cite{kot,ls99}.  It is  believed that  whether the spectrum is pure point or singular continuous     depends  on  resonances arising from the small denominators   and the  Lyapunov exponent.

There are two types of resonances appearing naturally.  The first one, the frequency resonance, was first observed in \cite{gordon,as82}, which is    quantified by 
\begin{equation}\label{gresf}
\beta(\alpha)=\limsup_{|k|\rightarrow \infty}-\frac{\ln ||k\alpha||_{\R/\Z}}{|k|},
\end{equation}
where  $ ||x||_{\R/\Z}={\rm dist} (x,\Z) $.

Another  one,  the phase resonance, happens for even functions \cite{js94} and is quantified by 
\begin{equation}\label{gresp}
\delta(\alpha,\theta)=\limsup_{|k|\to \infty}-\frac{\ln ||2\theta+k\alpha||_{\R/\Z}}{|k|}.
\end{equation}


%

A celebrated result of Jitomirskaya show that when $\beta(\alpha)=\delta(\alpha,\theta)=0$,   $H_{\lambda,\alpha,\theta} $ has  Anderson localization  for $|\lambda|>1$~\cite{jks05,j}. 
In~\cite{j94}, 	Jitomirskaya  conjectured  that 
for  the supercritical almost Mathieu operator,  
the spectral type (pure point spectrum/localization or singular continuous spectrum) completely determined by competitions between the Lyapunov exponent and  frequency/phase resonances quantified by  \eqref{gresf} and \eqref{gresp}. 
More precisely, 

\vspace{0.1in}

{\bf Conjecture 1}  
\begin{itemize}
	\item  the almost Mathieu operator  $H_{\lambda,\alpha,\theta}$ satisfies Anderson localization  if  $  |\lambda| > e^{\beta(\alpha)+\delta(\alpha,\theta)}$;
	\item the almost Mathieu operator  $H_{\lambda,\alpha,\theta}$  has purely singular continuous spectrum if  $ 1 <|\lambda| < e^{\beta(\alpha)+\delta(\alpha,\theta)}$.
\end{itemize}

Conjecture 1 
 is referred to as the second spectral transition line conjecture. 

Significant progress in establishing the second spectral transition line  has been made in the past 15 years ~\cite{lyjfa,liuetds20,lyjst,ayz,aj09,jl18,jl19,jk16,aj10,lyjfg,ajz18,geyou1,geyouzhao}. Here is a summary.
	
	{\bf Second spectral transition line}
	\begin{itemize}
		\item[1.] When $  |\lambda|>e^{\beta(\alpha)}$ and $\delta(\alpha,\theta)=0$,  $H_{\lambda,\alpha,\theta} $ has   Anderson localization 
		~\cite{jl18}.
		\item[2.]  When $ |\lambda|>e^{\delta(\alpha,\theta)}$ and $\beta(\alpha)=0$,  $H_{\lambda,\alpha,\theta} $ has   Anderson localization 
		~\cite{jl19}.
		\item[3.] When $|\lambda|>e^{2\beta(\alpha)}$,  
		$H_{\lambda,\alpha,\theta}$
		satisfies Anderson localization  for completely  resonant phase $\theta$, namely,   $2\theta\in \alpha \mathbb{Z}+\mathbb{Z}$  ~\cite{l21}.  
		\item[4.] When $1< |\lambda|<e^{\beta(\alpha)}$,   $H_{\lambda,\alpha,\theta} $ has purely singular continuous spectrum ~\cite{ayz}.
		
		\item[5.]  When $1< |\lambda|<e^{\delta(\alpha,\theta)}$,   $H_{\lambda,\alpha,\theta} $ has purely singular continuous spectrum ~\cite{jl19}.
		
	\end{itemize}
	It is worthwhile to mention that 
	Part 3 solves   a conjecture of 	Avila and Jitomirskaya ~\cite{aj09}. 
	It is easy to see that 
	for completely  resonant phases $\theta$,  namely $2\theta\in\alpha\Z+\Z$,  $\delta(\alpha,\theta)=\beta(\alpha)$,  so Part 3 is a particular case of Conjecture 1.  
	All previous results give positive answers to Conjecture 1.

In this paper,
we discover that not just strength of frequency and phase resonances  plays a role in the supercritical regime,  
distributions of  
 resonances are also crucial to the spectral types.
	In particular, we  disprove Conjecture 1. 
	
	We postpone our main theorem to Section 2 and list  some  corollaries first.
	 
	We denote the  continued fraction expansion  of $\alpha$ by
	\begin{equation}\label{g-1}
		\alpha=[a_0,a_1,a_2,\cdots]=a_0+\frac{1}{a_1+\frac{1}{a_2+\frac{1}{a_3+\frac{1}{\ddots}}}}.
	\end{equation}

	
	\begin{corollary}\label{main1}
		Assume that $\alpha$  satisfies  $0<\beta(\alpha)<\infty$  and $\liminf_{n\to\infty} a_n\geq 10^{20}$.  
		Then there exists  $\theta$  with $0<\delta(\alpha,\theta)<\infty$ such that 
		\begin{itemize}
			\item[1.] for $|\lambda|>e^{ \delta(\alpha,\theta)}$, 	$H_{\lambda,\alpha,\theta}$ has Anderson localization;
			\item[2.]  for $1\leq  |\lambda| <e^{ \delta(\alpha,\theta)}$, 	$H_{\lambda,\alpha,\theta}$ has  purely singular continuous spectrum;
			\item[3.] for $ |\lambda| <1$, 	$H_{\lambda,\alpha,\theta}$ has  purely absolutely continuous spectrum.
		\end{itemize}
		
	\end{corollary}

\begin{corollary}\label{cor1}
		For any $0<\mu<\infty$, 
	there exist $\alpha$ with $\beta(\alpha)=\mu$ and  $\theta$  with $0<\delta(\alpha,\theta)<\infty$  such that 
	\begin{itemize}
		\item[1.] for $|\lambda|>e^{ \delta(\alpha,\theta)}$, 	$H_{\lambda,\alpha,\theta}$ has Anderson localization;
		\item[2.]  for $1\leq  |\lambda| <e^{ \delta(\alpha,\theta)}$, 	$H_{\lambda,\alpha,\theta}$ has  purely singular continuous spectrum;
		\item[3.] for $ |\lambda| <1$, 	$H_{\lambda,\alpha,\theta}$ has  purely absolutely continuous spectrum.
	\end{itemize}
\end{corollary}

	\begin{remark}
		\begin{itemize}\label{re}
				\item  
			The second spectral transition line for    $(\alpha,\theta)$ in     Corollaries \ref{main1} and \ref{cor1}   is $|\lambda|=e^{\delta(\alpha,\theta)}$ differing from $|\lambda|=e^{\beta(\alpha)+\delta(\alpha,\theta)}$ , which disproves  Conjecture 1.
				\item 
				Parts 2 and 3  in  Corollaries \ref{main1} and \ref{cor1} are known  for all possible $\alpha$ and $\theta$ ~\cite{jl19,j19,avi08}. They are  included here for completeness.
				
			\item  Part 1 in Corollaries \ref{main1} and \ref{cor1}   follows from our main Theorem \ref{main2} below and some simple facts about continued fraction expansions.

			
		\end{itemize}
		
	\end{remark}
	
The proof of  our main results is constructive. 
 As we mentioned earlier,  our  key observation is that  the  distribution of resonances   plays a role in determining  the spectral type.  The main idea behind our construction is that when    locations of phase resonances and  frequency resonances are  far away, the Lyapunov exponent with a smaller magnitude  (predicted by Conjecture 1) is sufficient to beat both frequency and phase resonances.  Our analysis in the present paper gives us the indication that Conjecture 1 is only true  for the two extreme cases 
 ($\beta(\alpha)=0$ or $\delta(\alpha,\theta)=0$). 
 So,
 we conjecture that 
 
{\bf Conjecture 2}  

For any $\alpha$ with $0<\beta(\alpha)<\infty$ and $0<\kappa<\infty$,  there exists $\theta$  with 
 $\delta(\alpha,\theta)=\kappa$ such that the following is not true. 
$H_{\lambda,\alpha,\theta}$ has Anderson localization  for $   |\lambda| > e^{\kappa+\beta(\alpha)}$  and  has purely singular continuous  spectrum for  $1<|\lambda| <e^{\kappa+\beta(\alpha)}$.

	We list two  simpler versions of Conjecture 2 here.
	
	{\bf Conjecture 3}  
	
For any $\alpha$ with $0<\beta(\alpha)<\infty$,  there exists $\theta$  with 
$0<\delta(\alpha,\theta)<\infty$ such that the following is not true. 
$H_{\lambda,\alpha,\theta}$ has Anderson localization  for $   |\lambda| > e^{\delta(\alpha,\theta)+\beta(\alpha)}$  and  has purely singular continuous  spectrum for  $1<|\lambda| <e^{\delta(\alpha,\theta)+\beta(\alpha)}$.
	
	{\bf Conjecture 4}  
	
	For any  $0<\mu<\infty$ and $0<\kappa<\infty$,  there exist  $\alpha$ and $\theta$  with 
	$\beta(\alpha)=\mu$ and $\delta(\alpha,\theta)=\kappa$ such that the following is not true. 
	$H_{\lambda,\alpha,\theta}$ has Anderson localization  for $   |\lambda| > e^{\kappa+\mu}$  and  has purely singular continuous  spectrum for  $1<|\lambda| <e^{\kappa+\mu}$.
	
	As a corollary of Theorem \ref{main2} below,  we 
	prove Conjecture 4 in the regime $  \kappa \in[ 100 \mu,\infty)$ and $0<\mu<\infty$, namely,  
	
	\begin{corollary}\label{cor2}
For any	$0<\mu<\infty$ and any $\kappa \in[100 \mu,\infty)$,  
		there exist $\alpha$ and $\theta$ with $\beta(\alpha)=\mu$ and    $\delta(\alpha,\theta)=\kappa$  such that 
		\begin{itemize}
			\item[1.] for $|\lambda|>e^{ \delta(\alpha,\theta)}$, 	$H_{\lambda,\alpha,\theta}$ has Anderson localization;
			\item[2.]  for $1\leq  |\lambda| <e^{ \delta(\alpha,\theta)}$, 	$H_{\lambda,\alpha,\theta}$ has  purely singular continuous spectrum;
			\item[3.] for $ |\lambda| <1$, 	$H_{\lambda,\alpha,\theta}$ has  purely absolutely continuous spectrum.
		\end{itemize}
	\end{corollary}

	\section{Constructions  of  $\theta$}
	
	
	Without loss of generality, assume 
	$\alpha\in [0,1]\backslash \Q$.  Recall \eqref{g-1} and 
for $n\geq 1$, let 
	\begin{equation}\label{g101}
		\frac{p_n}{q_n}=[a_1,a_2,\cdots,a_n]=\frac{1}{a_1+\frac{1}{a_2+\frac{1}{a_3+\frac{1}{\ddots+\frac{1}{a_n}}}}},
	\end{equation}
	where $(p_n,q_n)=1$. 
	By setting $q_0=1$ and $q_{-1}=0$,  one has that  for any $n\geq 1$,
	\begin{equation}\label{g1}
		q_{n}=a_n q_{n-1}+q_{n-2}.
	\end{equation}

	
	By  basic facts  of the continued fraction expansion,  we have that  	for any $1\leq k<q_{n+1}$,
	\begin{equation}\label{GDC1}
		||k\alpha||_{\R/\Z}\geq  ||q_n\alpha||_{\R/\Z},
	\end{equation}
	\begin{equation}\label{GDC2}
		\frac{1}{q_n+q_{n+1}}\leq ||q_n\alpha||_{\R/\Z}=|q_n\alpha -p_n|\leq\frac{1}{q_{n+1}}.
	\end{equation}
	and for even $n$,
	\begin{equation}\label{GDC3}
		\frac{p_{n}}{q_{n}}< \alpha< \frac{p_{n+1}}{q_{n+1}}.
	\end{equation}
	
	By \eqref{gresf} and \eqref{g1}-\eqref{GDC2}, one has that 
	\begin{equation}\label{g3}
		\beta(\alpha)=\limsup_{n\to\infty}\frac{\ln q_{n+1}}{q_n}. 
	\end{equation}
	We are in a position to construct $\theta$.

	Fix any $\eta$ with $0<\eta\leq  10^{-2}$.
	Assume that $\alpha$ satisfies 
	\begin{equation}\label{gnew1}
		\liminf_{n\to\infty} a_n\geq 10^{18}\eta^{-1}. 
	\end{equation}
	
	By the assumption \eqref{gnew1}, there exists  large enough
	$ \tilde{j}_0$  such that  for any $j\geq \tilde{j}_0$,
	\begin{equation}\label{gj22}
 \eta q_{j+1}\geq 10^{17}q_{j}.
	\end{equation}
	

	Let $k_{0}$  be a positive integer  (the existence is guaranteed by the ergodic theory) so that  
	\begin{equation}\label{gj21}
(-k_{0}\alpha)\mod\Z \in \left[\frac{1}{10},\frac{1}{5}\right].
	\end{equation}
	
	
{\bf Case 1:} $\limsup_{j\to \infty}\frac{\ln q_{2j}}{q_{2j-1}}=\beta(\alpha)$.

Let $k_{ \tilde{j}_0}=k_0$.
For  any $j\geq \tilde{j}_0$, define
\begin{equation}\label{g4}
	k_{j+1}=k_{j}+\left\lfloor \eta \frac{q_{2j+1} }{q_{2j}} \right\rfloor  q_{2j},
\end{equation}
where  $\lfloor x \rfloor$ is the largest integer that is less than or equal to $x$.
For  any  $j\geq  \tilde{j}_0$, let 
\begin{equation}\label{g0}
	I_j=\left \{\theta\in [0,1/2]:-	 \frac{10\eta}{q_{2j}}\leq 2\theta- ((-k_j \alpha)\mod\Z) \leq   -\frac{\eta}{10q_{2j}}
	\right\}.
\end{equation}

By \eqref{GDC2}, \eqref{GDC3}, \eqref{gj22},  \eqref{gj21},  \eqref{g4} and inductions, one has that  for  any $j\geq  \tilde{j}_0$,
\begin{align*}
(	(-k_{j+1} \alpha)\mod \Z)-(	(-k_{j} \alpha)\mod \Z)&=-\left\lfloor \eta \frac{q_{2j+1} }{q_{2j}} \right\rfloor (q_{2j} \alpha -p_{2j})\\
	&\geq- \eta \frac{q_{2j+1} }{q_{2j}} |q_{2j} \alpha -p_{2j}|\\
	&\geq  -\frac{\eta }{q_{2j}},
\end{align*}
\begin{align*}
(	(-k_{j+1} \alpha)\mod \Z)-(	(-k_{j} \alpha)\mod \Z)
	&\leq - \eta \frac{q_{2j+1} }{2q_{2j}}|q_{2j} \alpha -p_{2j}| \\
	&\leq- \frac{\eta }{4q_{2j}},
\end{align*}
and
\begin{equation*}
(-k_j \alpha)\mod\Z \in \left[\frac{1}{20},\frac{1}{5}\right].
\end{equation*}
Therefore, we have that    $I_j$ (a closed interval),  $j\geq  \tilde{j}_0$,  is monotone, namely $I_{j+1}\subset I_j$. 
This implies  there exists  $\theta$   such that 
\begin{equation}\label{g5}
	\theta = \bigcap_{j= \tilde{j}_0}^{\infty} I_j.
\end{equation}

{\bf Case 2:} $\limsup_{j\to \infty}\frac{\ln q_{2j+1}}{q_{2j}}=\beta(\alpha)$.

For  any $j\geq \tilde{j}_0$, define
\begin{equation}\label{g4case2}
	k_{j+1}=k_{j}+\left\lfloor \eta \frac{q_{2j} }{q_{2j-1}} \right\rfloor  q_{2j-1} .
\end{equation}

For  any  $j\geq  \tilde{j}_0$, let 
\begin{equation}\label{g0case2}
	I_j=\left \{\theta\in [0,1/2]: \frac{\eta}{10q_{2j-1}} \leq 2\theta- ((-k_j \alpha) \mod \Z) \leq    \frac{10\eta}{q_{2j-1}}
	\right\}.
\end{equation}

By \eqref{GDC2}, \eqref{GDC3}, \eqref{gj22},  \eqref{gj21},  \eqref{g4case2} and inductions, one has that  for  any $j\geq  \tilde{j}_0$,
\begin{align*}
(	(-k_{j+1} \alpha)\mod \Z)-(	(-k_{j} \alpha)\mod \Z)&=-\left\lfloor \eta \frac{q_{2j} }{q_{2j-1}} \right\rfloor ( q_{2j-1} \alpha -p_{2j-1} )\\
	&\leq\eta \frac{q_{2j} }{q_{2j-1}}  | q_{2j-1} \alpha -p_{2j-1} | \\
	&\leq \frac{\eta }{q_{2j-1}},
\end{align*}
 
\begin{align*}
(	(-k_{j+1} \alpha)\mod \Z)-(	(-k_{j} \alpha)\mod \Z)
	&\geq  \eta \frac{q_{2j} }{2q_{2j-1}} | q_{2j-1} \alpha -p_{2j-1} | \\
	&\geq \frac{\eta }{4q_{2j-1}},
\end{align*}
and
\begin{equation*}
(-k_j \alpha)\mod\Z \in \left[\frac{1}{10}, \frac{2}{5} \right].
\end{equation*}
Therefore, we have that    $I_j$ (a closed interval),  $j\geq  \tilde{j}_0$, is monotone, namely $I_{j+1}\subset I_j$. This implies  there exists  $\theta$   such that 
\begin{equation}\label{g5case2}
	\theta = \bigcap_{j= \tilde{j}_0}^{\infty} I_j.
\end{equation}

\begin{theorem}\label{main2}
		Fix any $\eta$ with $0<\eta\leq  10^{-2}$.
		Assume that $\alpha$  satisfies  $0<\beta(\alpha)<\infty$  and $\liminf_{n\to\infty} a_n\geq 10^{18}\eta^{-1}$.  
	Let  $\theta$ be constructed  by \eqref{g5} or \eqref{g5case2}. Then   
	for any $|\lambda|>e^{ \delta(\alpha,\theta)}$, 	$H_{\lambda,\alpha,\theta}$ has Anderson localization.

\end{theorem}

%
%

In   the following sections, we prove Theorem \ref{main2}.  In order to avoid  repetitions, we only prove the Case 1:  $\limsup_{j\to \infty}\frac{\ln q_{2j}}{q_{2j-1}}=\beta(\alpha)$. The proof  of Case 2 is similar. 
We list the definitions and standard facts  in Section \ref{S3}. 
In Section \ref{S4}, we provide several technical lemmas.
Sections \ref{S5}  and \ref{S6} are devoted to treating non-resonant and  resonant sites respectively.  In Section \ref{Sclaim}, we prove corollaries and three claims which are used in Sections  \ref{S5}  and \ref{S6}.  Proofs of those claims are quite standard if readers are familiar with the small denominator arguments related to almost Mathieu operators.

 \section{Basics}\label{S3}
All basics in this section are taken from ~\cite{liuetds20}. We refer readers to ~\cite{liuetds20} and references therein for details.
Let $H$ be an operator on $\ell^2(\Z)$.
We say $\phi$ is a generalized eigenfunction corresponding to the generalized eigenvalue $E$ if 
\begin{equation}\label{g68}
	H\phi=E\phi  ,\text{ and }  |\phi(k)|\leq \hat{C}( 1+|k|).
\end{equation}
By Shnol's theorem,  in order to
prove  Anderson localization of   $H$,  we only need  to show that  every  generalized eigenfunction  is in fact an exponentially decaying eigenfunction, namely,  there exists some constant $c>0$ such that
\begin{equation*}
	| \phi(k)|\leq  e^{-c|k|} \text{ for large } k.
\end{equation*}	
For simplicity, we assume    $\hat{C}=1$ in \eqref{g68}.


From now on, we always assume $\phi$  is   a  generalized  eigenfunction  of $H_{\lambda,\alpha,\theta}$ and $E$ is the corresponding generalized eigenvalue.  Without loss of generality assume $\phi(0)=1$. It is well known that every generalized eigenvalue must be in the spectrum, namely $E\in \Sigma_{\lambda,\alpha}$, where
$\Sigma_{\lambda,\alpha}$ is  the spectrum of   $H_{\lambda,\alpha,\theta}$ (the spectrum  does not depend on $\theta$).   

For any $x_1,x_2\in\Z$ with $x_1<x_2$, denote by 
\begin{equation*}
	P_{[x_1,x_2]}(\lambda,\alpha,\theta)	=\det(R_{[x_1,x_2]}(H_{\lambda,\alpha,\theta}-E) R_{[x_1,x_2]}),
\end{equation*}	
where $R_{[x_1,x_2]}$ is the restriction on $[x_1,x_2]$. 
Let us denote
$$ P_k(\lambda,\alpha,\theta)=\det(R_{[0,k-1]}(H_{\lambda,\alpha,\theta}-E) R_{[0,k-1]}).$$
When there is no ambiguity, we drop the  dependence on  parameters $E,\lambda,\alpha$ or $ \theta$.
Clearly,
\begin{equation}\label{g108}
	P_{[x_1,x_2]}(\theta)	=P_{k} (\theta+x_1\alpha),
\end{equation}
where $k=x_2-x_1+1$.



Let
\begin{equation}\label{G.transfer}
	A_{k}(\theta)=\prod_{j=k-1}^{0 }A(\theta+j\alpha)=A(\theta+(k-1)\alpha)A(\theta+(k-2)\alpha)\cdots A(\theta)
\end{equation}
and
\begin{equation}\label{G.transfer1}
	A_{-k}(\theta)=A_{k}^{-1}(\theta-k\alpha)
\end{equation}
for $k\geq 1$,
where $A(\theta)=\left(
\begin{array}{cc}
	E-2\lambda\cos2\pi\theta & -1 \\
	1& 0\\
\end{array}
\right)
$.
$A_{k}$  is called the ($k$-step) transfer matrix. 

By the definition,   for  any $k\in\Z_+, m\in\Z$, one has that
\begin{equation}\label{G.new17}
	\left(\begin{array}{c}
		\phi(k+m) \\
		\phi(k+m-1)                                                                                        \end{array}\right)
	=A_{k}(\theta+m\alpha)
	\left(\begin{array}{c}
		\phi(m) \\
		\phi(m-1)                                                                                        \end{array}\right).
\end{equation}

It is easy to check  that for $k\in\Z_+$, 
\begin{equation}\label{G34}
	A_{k}(\theta)=
	\left(
	\begin{array}{cc}
		P_k(\theta) &- P_{k-1}(\theta+\alpha)\\
		P_{k-1}(\theta) & - P_{k-2}(\theta+\alpha) \\
	\end{array}
	\right).
\end{equation}
The Lyapunov exponent 
is given  by
\begin{equation}\label{G21}
	L(E)=\lim_{k\rightarrow\infty} \frac{1}{k}\int_{\mathbb{R}/\mathbb{Z}} \ln \| A_k(\theta)\|d\theta.
\end{equation}
The Lyapunov exponent can be computed precisely for $E$ in the
spectrum of $H_{\lambda,\alpha,\theta}$. 
\begin{lemma}~\cite{bj02}\label{lya}
	For $E\in \Sigma_{\lambda,\alpha}$ and $|\lambda|>1$, we have
	$L(E)=\ln|\lambda|$.
\end{lemma}
In the following, let
$L := \ln|\lambda| $ be the Lyapunov exponent of the almost Mathieu operator for energies in the spectrum.

By
upper semicontinuity and unique ergodicity,
one has
\begin{equation}\label{G23}
	L=\lim_{k\rightarrow\infty} \sup_{\theta\in\mathbb{R}}\frac{1}{k} \ln \| A_k(\theta)\|.
\end{equation}
Therefore, 
for any  $  \varepsilon >0$,  
\begin{equation}\label{G24}
	\| A_k(\theta)\|\leq  e^{(L+\varepsilon)k},
\end{equation}
when  $k$ is large enough  (independent of $\theta$).

By \eqref{G34} and \eqref{G24},  one has that  for large $k$,
\begin{equation}\label{Numerator}
	| P_k(\theta)|\leq  e^{(L+\varepsilon)k},
\end{equation}
and hence for large $|x_2-x_1|$,
\begin{equation}\label{Numerator1}
	| P_{[x_1,x_2]}(\theta)|\leq  e^{(L+\varepsilon)|x_2-x_1|}.
\end{equation}
By \eqref{G.new17} and \eqref{G24}, one has that  for large $|k_1-k_2|$,
\begin{equation}\label{g500}
	\left \|	\left(\begin{array}{c}
		\phi(k_1+1) \\
		\phi(k_1)                                                                                        \end{array}\right)\right\|
	\leq  e^{(L+\varepsilon) |k_1-k_2|}\left\|
	\left(\begin{array}{c}
		\phi(k_2+1) \\
		\phi(k_2)                                                                                        \end{array}\right)\right\|.
\end{equation}
For any $x_1,x_2\in\Z$ with $x_1<x_2$,  let $G_{[x_1,x_2]}$  be the Green's function:
\begin{equation*}
	G_{[x_1,x_2]}=(R_{[x_1,x_2]} (H_{\lambda,\alpha,\theta}-E) R_{[x_1,x_2]})^{-1}.
\end{equation*}
By Cramer's rule,   for  any
$ y\in [x_1,x_2]  $,  one has
\begin{eqnarray}
	|G_{[x_1,x_2]}(x_1,y)| &=&  \left| \frac{P_{[y+1,x_2]}}{P_{[x_1,x_2]}}\right|,\label{Cramer1}\\
	|G_{[x_1,x_2]}(y,x_2)| &=&\left|\frac{P_{[x_1,y-1]}}{P_{[x_1,x_2]}} \right|.\label{Cramer2}
\end{eqnarray}
It is  easy to check that 
\begin{equation}\label{Block}
	\phi(y)= -G_{[x_1,x_2]}(x_1,y ) \phi(x_1-1)-G_{[x_1,x_2]}(y,x_2) \phi(x_2+1).
\end{equation}

Denote by   $x_1^\prime=x_1-1$ and $x_2^\prime=x_2+1$.

By \eqref{Cramer1}, \eqref{Cramer2} and \eqref{Block}, 
one has that  for any $y\in [x_1,x_2]$, 
\begin{equation}\label{g100}
	|\phi(y)|\leq |P_{[x_1,x_2]}|^{-1} |P_{[x_1,y-1]}| |\phi(x_2^{\prime})| + |P_{[x_1,x_2]}|^{-1} |P_{[y+1,x_2]}| |\phi(x_1^{\prime})|.
\end{equation}

Given  a set $\{\theta_1, \cdots ,\theta_{k+1}\}$,  the Lagrange interpolation terms $\text{Lag}_m$, $m=1,2,\cdots,k+1$, are  defined by
\begin{equation}\label{Def.Uniform}
	\text{Lag}_m= \ln \max_{ x\in[-1,1]} \prod_{ j=1 , j\neq m }^{k+1}\frac{|x-\cos2\pi\theta_j|}
	{|\cos2\pi\theta_m-\cos2\pi\theta_j|}.
\end{equation}

\begin{lemma}
	\label{Le.Uniform}
	Given a set $\{\theta_1, \cdots ,\theta_{k+1}\}$,   there exists some $\theta_m$ in set  $\{\theta_1, \cdots ,\theta_{k+1}\}$ such that
	
	\begin{equation*}
		\left	|	P_{k}\left(\theta_m -\frac{k-1}{2}\alpha\right)\right|\geq \frac{e^{k L-\rm{Lag}_m}}{k+1}.
	\end{equation*}
	
\end{lemma}

Assume $k$ is odd.
Set $I=[j_0-\frac{k+1}{2}+1,j_0+\frac{k+1}{2}-1]=[x_1,x_2]$.  Assume that  $ |P_{[x_1,x_2]}|\geq e^{ k L - \mu}$.
By  \eqref{Numerator1}  and  \eqref{g100}, for any $y\in [x_1,x_2]$ with large enough $|y-x_1|$ and $|y-x_2|$,
one has that
\begin{align}
	|\phi(y)| & \leq  \sum_{i=1,2} e^{(L+\varepsilon )(k-|y-x_i|)-kL +\mu}|\phi(x_i^\prime)|\nonumber\\
	&   \leq e^{\mu} \sum_{i=1,2} | \phi(x_i^{\prime})|e^{\varepsilon k-L|y-x_i| }.\label{g103}
\end{align}

\section{Technical preparations }\label{S4}
The  estimates in the following proposition will be constantly used in the proof.
\begin{proposition}\label{prop}
	Let $\theta$ be given by \eqref{g5}.
	When  $j$ is large enough, 
	 we have that 
	\begin{enumerate}
		\item  for any $ k=k_j+l_1 q_{2j-1}$ with $1\leq |l_1|\leq \frac{ q_{2j}}{10}$, 
		\begin{equation}\label{gdc1}
			||2\theta+k\alpha||_{\R/\Z}\geq \frac{1}{4q_{2j}};
		\end{equation}
		\item   for any $ k=k_j+l_1 q_{2j-1}+l_2$ with $|l_1|\leq \frac{ q_{2j}}{10q_{2j-1} }$ and $1\leq   l_2<q_{2j-1}$, 
		\begin{equation}\label{gdc2}
			||2\theta+k\alpha||_{\R/\Z} \geq \frac{1}{4q_{2j-1}};
		\end{equation}
		\item  for any  $ k=k_j+l_1 q_{2j}$ with $|l_1|\leq \frac{\eta q_{2j+1}}{30q_{2j} }$,
		\begin{equation}\label{gdc3}
			||2\theta+k\alpha||_{\R/\Z} \geq \frac{\eta}{20q_{2j}};
		\end{equation}
		\item    for any $ k=k_j+l_1 q_{2j}+l_2$ with $|l_1|\leq \frac{10\eta q_{2j+1}}{q_{2j} }$ and $1\leq  l_2<q_{2j}$, 
		\begin{equation}\label{gdc4}
			||2\theta+k\alpha||_{\R/\Z} \geq \frac{1}{4q_{2j}};
		\end{equation}
		\item  for any  $k=l_1 q_n+l_2$ with $|l_1|\leq \frac{q_{n+1}}{10q_n}$ and $1\leq l_2<q_n$,
		\begin{equation}\label{gdc5}
			|| k\alpha||_{\R/\Z} \geq \frac{1}{4q_{n}}.
		\end{equation}
	\end{enumerate}
\end{proposition}
	\begin{proof}
		The proof is based on \eqref{GDC1}, \eqref{GDC2}  and \eqref{g0}.
		By 	direct computations, 
		we have that   in (1),
		\begin{align*}
			||2\theta+k\alpha||_{\R/\Z} &= ||2\theta+k_{j}\alpha+l_1q_{2j-1}\alpha ||_{\R/\Z} \\
			&\geq   |l_1|\;||q_{2j-1}\alpha||_{\R/\Z}-|| 2\theta+k_{j}\alpha||_{\R/\Z} \\
			&\geq \frac{|l_1|}{2q_{2j}} -\frac{10\eta}{q_{2j}}\\
			&\geq \frac{1}{4q_{2j}},
		\end{align*}	
		in (2)
		\begin{align*}
			||2\theta+k\alpha||_{\R/\Z} &= ||2\theta+k_{j}\alpha+l_1q_{2j-1}\alpha+l_2\alpha||_{\R/\Z} \\
			&\geq   || l_2\alpha||_{\R/\Z} -||2\theta+k_{j}\alpha||_{\R/\Z}-||l_1q_{2j-1}\alpha||_{\R/\Z}\\
			&\geq\frac{1}{2q_{2j-1}}-\frac{10\eta}{q_{2j}}-\frac{|l_1|}{q_{2j}}\\
			&\geq \frac{1}{2q_{2j-1}}-\frac{10\eta}{q_{2j}}-\frac{1}{10q_{2j-1}}\\
			&\geq \frac{1}{4q_{2j-1}},
		\end{align*}	
		in (3),
		\begin{align*}
			||2\theta+k\alpha||_{\R/\Z} &= ||2\theta+k_{j}\alpha+l_1q_{2j}\alpha ||_{\R/\Z} \\
			&\geq   || 2\theta+k_{j}\alpha||_{\R/\Z} -|l_1|\;||q_{2j}\alpha||_{\R/\Z}\\
			&\geq\frac{\eta}{10q_{2j}}-\frac{|l_1|}{q_{2j+1}}\\
			&\geq \frac{\eta}{20q_{2j}},
		\end{align*}
		in (4)
		\begin{align*}
			||2\theta+k\alpha||_{\R/\Z} &= ||2\theta+k_{j}\alpha+l_1q_{2j}\alpha+l_2\alpha||_{\R/\Z} \\
			&\geq   || l_2\alpha||_{\R/\Z} -||2\theta+k_{j}\alpha||_{\R/\Z}-||l_1q_{2j}\alpha||_{\R/\Z}\\
			&\geq\frac{1}{2q_{2j}}-\frac{10\eta}{q_{2j}}-\frac{|l_1|}{q_{2j+1}}\\
			&\geq \frac{1}{2q_{2j}}-\frac{20\eta}{q_{2j}}\\
			&\geq \frac{1}{4q_{2j}},
		\end{align*}	
		and in (5),
		\begin{align*}
			|| k\alpha||_{\R/\Z} &= || l_1 q_n\alpha+l_2 \alpha||_{\R/\Z} \\
			&\geq   || l_2\alpha||_{\R/\Z} -||l_1q_{n}\alpha||_{\R/\Z}\\
			&\geq\frac{1}{2q_{n}}- \frac{|l_1|}{q_{n+1}}\\
			&\geq \frac{1}{4q_{n}}.
		\end{align*}

	\end{proof}

\begin{lemma}\label{le0}
	Let $\theta$ be defined by \eqref{g5}. Then
	\begin{equation}\label{g62}
	\delta(\alpha,\theta)= \limsup_{j\to \infty}-\frac{\ln ||2\theta+k_j\alpha||_{\R/\Z}}{k_j}=  \limsup_{j\to \infty}\frac{\ln q_{2j}}{k_j}
	\end{equation}
and 
	\begin{equation}\label{g6}
	 (1-10^{-14})\frac{\beta}{\eta}\leq 	\delta(\alpha,\theta)\leq (1+10^{-14})\frac{\beta}{\eta}.
	\end{equation}
\end{lemma}
\begin{proof}
	By \eqref{g4}, one has  that for $j\geq j_0$,
	\begin{equation}\label{gj31}
-q_{2j}	\leq k_{j+1} -k_j-\eta q_{2j+1}\leq 0.
	\end{equation}
	By \eqref{gj31} and inductions, we have that for  large $j$,
	\begin{equation}\label{g501}
		|k_j-\eta q_{2j-1}| \leq 100 q_{2j-2}\leq 10^{-15} \eta q_{2j-1}.
	\end{equation}
This implies  that
	\begin{align}
	\delta(\alpha,\theta)&=\limsup_{k\to \infty}-\frac{\ln ||2\theta+k\alpha||_{\R/\Z}}{|k|}\nonumber\\
	&\geq \limsup_{j\to \infty}-\frac{\ln ||2\theta+k_j\alpha||_{\R/\Z}}{k_j}\label{gj151}\\
	&\geq  (1-10^{-14})\frac{\beta}{\eta}.\label{gj152}
\end{align}

	Assume $k$ satisfies  $k_j<|k|\leq k_{j+1}$.
	
	{\bf Case 1:} $|k|\geq 10^{-3} \eta q_{2j+1}$ and $k\neq k_{j+1}$
	
	In this case, by \eqref{GDC1}, \eqref{GDC2}  and \eqref{g0}, one has 
	\begin{align*}
		||2\theta+k\alpha||_{\R/\Z} &= ||2\theta+k_{j+1}\alpha+(k-k_{j+1})\alpha||_{\R/\Z} \\
		&\geq   || (k-k_{j+1})\alpha||_{\R/\Z} -||2\theta+k_{j+1}\alpha||_{\R/\Z}\\
		&\geq ||q_{2j}\alpha||_{\R/\Z}-||2\theta+k_{j+1}\alpha||_{\R/\Z}\\
		&\geq \frac{1}{2q_{2j+1}}-\frac{10\eta}{q_{2j+2}}\\
		&\geq \frac{1}{4q_{2j+1}}.
	\end{align*}	
	
	{\bf Case 2:} $k_j<|k|< 10^{-3} \eta q_{2j+1}$
	
	Write $k$ as $k=k_j+l_1 q_{2j}+l_2$, where $l_1\in\Z$ with $ |l_1|\leq \frac{ 10^{-3} \eta q_{2j+1}}{q_{2j}}+2$ and $0\leq l_2<q_{2j}$.
	
	By \eqref{gdc3} and \eqref{gdc4}, one has that 
	\begin{equation*}
		||2\theta+k\alpha||_{\R/\Z} \geq \frac{\eta}{20q_{2j}}.
	\end{equation*}

	Putting both cases together,   we have that 
	\begin{align}
		\delta(\alpha,\theta)&=\limsup_{k\to \infty}-\frac{\ln ||2\theta+k\alpha||_{\R/\Z}}{|k|}\nonumber\\
		&\leq \limsup_{j\to \infty}\frac{\ln q_{2j}}{k_j} \label{gj153}\\
		&\leq  (1+10^{-14})\frac{\beta}{\eta}.\label{gj154}
	\end{align}
We conclude that   \eqref{g62} follows from \eqref{gj151} and \eqref{gj153},  and \eqref{g6}  follows from \eqref{gj152} and \eqref{gj154}.
\end{proof}
\begin{lemma}\label{le01}
	Let $\theta$ be defined by \eqref{g5}. Assume $ \liminf a_n=\infty$.  Then
	\begin{equation}\label{g61}
	 	\delta(\alpha,\theta)=\frac{\beta}{\eta}.
	\end{equation}
\end{lemma}
\begin{proof}
By the first inequality of \eqref{g501} and the fact that  $ \liminf a_n=\infty$, one has that 
\begin{equation}\label{gj32}
\lim_{j\to\infty} \frac{k_j}{q_{2j-1}}=\eta.
\end{equation}
Now \eqref{g61} follows from \eqref{g62} and \eqref{gj32}.
\end{proof}
\section{Non-resonant sites} \label{S5}

In the following,  we always assume 
\begin{itemize}
	\item $\theta$ is given by \eqref{g5},
	\item $L=\ln |\lambda|> \delta(\alpha,\theta)$,
	\item $\varepsilon>0$ is an arbitrarily small constant. 
	\item  $C$ is a large constant (depends on $\lambda$ and $\alpha$) and  it may change even in the same equation,
	\item $n$ is large enough  which depends on all  parameters and constants.
\end{itemize}
Let $b_n=10^{-7}\eta q_n$. 
For any $\ell \in\Z$, let
\begin{equation*}
	r_{\ell}^{\varepsilon,n}= \sup_{|r|\leq 10\varepsilon}|\phi(\ell q_n+rq_n)|,
\end{equation*}
and for $n=2j-1$, 
\begin{equation*}
	{r}_{\ell+\eta}^{\varepsilon,n}= \sup_{|r|\leq 10\varepsilon}|\phi(\ell q_n+k_j+rq_n)|.
\end{equation*}
From \eqref{g501}, one can see that   $ k_j\approx \eta q_n$. So   $	{r}_{\ell+\eta}^{\varepsilon,n}$ is essentially the value of $|\phi(y)|$ with $y\approx (\ell+\eta)q_n$. 
\begin{lemma}\label{le2}
	Assume that $|\lambda|>1$.
	Let $\ell$ be such that $0\leq |\ell| \leq  50\frac{b_{n+1}}{q_n}$. 
	Then for sufficiently large
	$n $, the following statements hold:
	\begin{itemize}
		\item 	when  $n=2j-1$ (namely, $n$ is odd),  we have that  for any 
		$y\in [\ell q_n+10\varepsilon q_n,\ell q_n+k_j -10\varepsilon q_n]$, 
		\begin{equation}\label{g11}
			|\phi(y)|\leq r_{\ell}^{\varepsilon,n}\exp\{-(L- \varepsilon)(|y-\ell q_n|-3\varepsilon q_n)\} + {r}_{\ell+\eta}^{\varepsilon,n}\exp\{-(L- \varepsilon)(|\ell q_n+k_j-y|-3\varepsilon q_n)\},
		\end{equation}
		and for   any 
		$y\in [\ell q_n+  k_j+10\varepsilon q_n, \ell q_n+q_n-10\varepsilon q_n]$, 
		\begin{equation}\label{g13}
			|\phi(y)|\leq r_{\ell+\eta}^{\varepsilon,n}\exp\{-(L- \varepsilon)(|y-\ell q_n- k_j|-3\varepsilon q_n)\} + {r}_{\ell+1}^{\varepsilon,n}\exp\{-(L- \varepsilon)(|\ell q_n+q_n-y|-3\varepsilon q_n)\};
		\end{equation}
		\item 
		when  $n=2j$ (namely, $n$ is even),  
		we have that  for any 
		$y\in [\ell q_n+10\varepsilon q_n,(\ell+ 1)q_n-10\varepsilon q_n]$, 
		\begin{equation}\label{g12}
			|\phi(y)|\leq r_{\ell}^{\varepsilon,n}\exp\{-(L- \varepsilon)(|y-\ell q_n|-3\varepsilon q_n)\} + {r}_{\ell+1}^{\varepsilon,n}\exp\{-(L- \varepsilon)(|\ell q_n+q_n-y|-3\varepsilon q_n)\}.
		\end{equation}
	\end{itemize}

\end{lemma}
\begin{proof}
	Assume $n=2j-1$ first. 
	For any $p\in [\ell q_n,\ell q_n+q_n]$ satisfying  $|p-\ell q_n|\geq  \varepsilon q_n$, $|p-\ell q_n-q_n|\geq   \varepsilon q_n$ and $|p-\ell q_n-k_j|\geq   \varepsilon q_n $, let
	\begin{equation*}
		d_p=\frac{1}{100}\min\{|p-\ell q_n|,|p-\ell q_n-q_n|,|p-\ell q_n- k_j|, 10^{-5} \eta q_n\}
	\end{equation*}
	Let $n_0$ be the  smallest integer such that
	\begin{equation*}
		2 q_{n-n_0} \leq  d_p,
	\end{equation*}
	and  let $s$ be the largest positive integer such that $2sq_{n-n_0}\leq
	d_p$. Notice that $ 2(s+1)q_{n-n_0}\geq d_p $, one has that
	\begin{equation}\label{g14}
		sq_{n-n_0}\geq \frac{1}{4}d_p\geq \frac{1}{400}\varepsilon q_n.
	\end{equation}
	{\bf Case 1}: $p \in [\ell q_n+\varepsilon q_n,\ell q_n+ k_j-\varepsilon q_n]$.

	If $p\leq \ell q_n+\frac{k_j}{2} $,
	we construct intervals
	\begin{equation*}
		I_1=[-2sq_{n-n_0},2s q_{n-n_0}-1],I_2=[p- 2sq_{n-n_0},p-1].
	\end{equation*}
	
	If $p>\ell q_n+\frac{k_j}{2} $, we construct intervals
	\begin{equation*}
		I_1=[- 2sq_{n-n_0} ,2sq_{n-n_0}-1],I_2=[p+1,p+2sq_{n-n_0}].
	\end{equation*}
	{\bf Case 2}: $p \in [\ell q_n+k_j+\varepsilon q_n,\ell q_n+q_n-\varepsilon q_n]$.

	If $p\leq \ell q_n+\frac{k_j}{2} +\frac{1}{2}q_n$,
	we construct intervals
	\begin{equation*}
		I_1=[-2sq_{n-n_0},2s q_{n-n_0}-1],I_2=[p- 2sq_{n-n_0},p-1].
	\end{equation*}
	
	If $p>\ell q_n+\frac{k_j}{2} +\frac{1}{2}q_n$, we construct intervals
	\begin{equation*}
		I_1=[- 2sq_{n-n_0} ,2sq_{n-n_0}-1],I_2=[p+1,p+2sq_{n-n_0}].
	\end{equation*}
	
	By the construction of $I_1$ and $I_2$, one has that for any $j_1,j_2\in I_1\cup I_2$  with $j_1\neq j_2$,
	there exist  $l_1$ and $l_2$ with $|l_1|\leq 100 \frac{b_{n+1}}{q_n}+4$ and $1\leq l_2<q_{n}$, and   $l_1'$ and $l_2'$ with  $|l_1'|\leq 100\frac{b_{n+1}}{q_n}+4$ and $1\leq l_2'<q_{n}$ such that 
	\begin{equation*}
		j_1-j_2=l_1 q_n+l_2,
	\end{equation*}
	and
	\begin{equation*}
		j_1+j_2=k_j+l_1' q_n+l_2'.
	\end{equation*}
	Therefore, by  \eqref{gdc2} and \eqref{gdc5}, we have that 
	that for any $j_1,j_2\in I_1\cup I_2$  with $j_1\neq j_2$,
	\begin{equation}\label{g201}
		||(j_1-j_2)\alpha||_{\R/\Z}\geq \frac{1}{4q_n}
	\end{equation}
	and
	\begin{equation}\label{g202}
		||2\theta+(j_1+j_2)\alpha||_{\R/\Z}\geq \frac{1}{4q_n}.
	\end{equation}

	Assume $n=2j$. 
	
	For any $p\in [\ell q_n,\ell q_n+q_n]$ satisfying  $|p-\ell q_n|\geq  \varepsilon q_n$ and $|p-\ell q_n-q_n|\geq   \varepsilon q_n$, let
	\begin{equation*}
		d_p=\frac{1}{100}\min\{|p-\ell q_n|,|p-\ell q_n-q_n|, 10^{-5} \eta q_n\}
	\end{equation*}
	Let $n_0$ be the  smallest integer such that
	\begin{equation*}
		2 q_{n-n_0} \leq  d_p,
	\end{equation*}
	and  let $s$ be the largest positive integer such that $2sq_{n-n_0}\leq
	d_p$. Notice that $ 2(s+1)q_{n-n_0}> d_p $, one has that
	\begin{equation}\label{g15}
		sq_{n-n_0}\geq \frac{1}{4}d_p\geq \frac{1}{400}\varepsilon q_n.
	\end{equation}
	We construct intervals
	\begin{equation*}
		I_1=[-sq_{n-n_0},s q_{n-n_0}-1],I_2=[p- sq_{n-n_0},p+s q_{n-n_0}-1].
	\end{equation*}

	In this case,  
	one has that for any $j_1,j_2\in I_1\cup I_2$  and $j_1\neq j_2$,
	\begin{equation*}
		j_1-j_2=l_1 q_n+l_2,
	\end{equation*}
	where $|l_1|\leq 100 \frac{b_{n+1}}{q_n}+4$ and $1\leq l_2<q_{n}$. Therefore, 
	by \eqref{gdc5}, we have that 
	\begin{equation}\label{g203}
		||(j_1-j_2)\alpha||_{\R/\Z}\geq \frac{1}{4q_n}.
	\end{equation}
	and by  \eqref{gdc3} and \eqref{gdc4}, we have that
	\begin{equation}\label{g204}
		||2\theta+(j_1+j_2)\alpha||_{\R/\Z}\geq \frac{\eta}{20q_n}.
	\end{equation}

	In all cases,  let $\theta_m=\theta+m\alpha$ for $m\in I_1\cup I_2$.  
	By the small divisor conditions  \eqref{g201}, \eqref{g202}, \eqref{g203} and \eqref{g204},  and modifying  the proof of ~\cite[Lemma 9.9]{aj09}  (or Appendices in ~\cite{jl18} and ~\cite{liuetds20}), we can prove  that
	for any
	$  m\in I_1\cup I_2$, one has  that ${\rm Lag}_m\leq  \varepsilon q_n  $.
	Now Lemma \ref{le2} follows from  standard block expansions (e.g., ~\cite[Lemma 4.1]{jl18} and ~\cite[Lemma 3.4]{jl19}). 
\end{proof}
\section{Resonant sites and proof of Theorem \ref{main2}}\label{S6}
We are going to deal with the resonances $\ell q_n+k_j $ and  $\ell q_n $ first.
For simplicity, we drop the dependence of $n$ and $\varepsilon$ in  notations  $r_{\ell}^{\varepsilon,n}$ and $r_{\ell+\eta}^{\varepsilon,n}$. Denote by 
 $\beta_n=\frac{\ln q_{n+1}}{q_n}$.  Clearly, for large $n$, $\beta_n\leq \beta+\varepsilon$. Since $L>\delta $, by Lemma \ref{le0}, one has that for large $j$, 
 \begin{equation}\label{gnew2}
 	L>\frac{\beta_{2j-1}q_{2j-1}}{k_j},
 \end{equation}
and for large $n$
 \begin{equation}\label{gnew3}
	L>(1-10^{-14})\frac{\beta}{\eta}> (1-10^{-13})\frac{\beta_n}{\eta}.
\end{equation}
\begin{theorem}\label{main3}
	Assume that $n=2j-1$ and $|\ell|\leq 20 \frac{b_{n+1}}{q_n} $. Then 
	\begin{equation}\label{g16}
		r_{\ell+\eta}\leq e^{C\varepsilon  q_n+\beta_n q_n} \left(e ^{ -L k_j} r_{\ell}+  e^{-Lq_n+Lk_j} r_{\ell+1}\right).
	\end{equation}
\end{theorem}
\begin{proof}
	Take any $p$ with $|p-(\ell q_n+k_j)|\leq 10\varepsilon q_n$   into consideration. 
	Let $n_0$ be the least positive integer such that
	\begin{equation*}
		q_{n-n_0}\leq  \frac{ \varepsilon}{2}  \left(\frac{k_j}{2}-2\varepsilon q_n \right) .
	\end{equation*}
	Let $s$ be the
	largest positive integer such that $sq_{n-n_0}\leq \frac{k_j}{2}-2 \varepsilon q_n $.
	By the fact that $(s+1)q_{n-n_0}\geq \frac{k_j}{2}-2\varepsilon q_n $, one has 
	\begin{equation*}
		s\geq \frac{1}{\varepsilon} 
	\end{equation*}
	and
	\begin{equation}\label{2sq}
	 \frac{k_j}{2}-3\varepsilon q_n \leq sq_{n-n_0}\leq \frac{k_j}{2}-2\varepsilon q_n.
	\end{equation}
	
	Construct intervals
	\begin{equation*}
		I_1=[-sq_{n-n_0},sq_{n-n_0}-1], I_2=[\ell q_n+k_j-sq_{n-n_0},\ell q_n+k_j+sq_{n-n_0}-1].
	\end{equation*}

	Let $\theta_m=\theta+m\alpha$ for $m\in I_1\cup I_2$. The set $\{\theta_m\}_{m\in I_1\cup I_2}$
	consists of $4sq_{n-n_0}$ elements. Let $k=4sq_{n-n_0}-1$.
	By  modifying  the proof of ~\cite[Lemma 9.9]{aj09} and   ~\cite[Lemma 4.1]{lyjfa} (or Appendices in ~\cite{jl18} and ~\cite{liuetds20}), we can prove the claim (Claim 1):
	for any
	$  m\in I_1\cup I_2$, one has ${\rm Lag}_m\leq  (\beta_n +\varepsilon) q_n  $.
	For  convenience, we  include a proof in Section \ref{Sclaim}.
	
	Applying   Lemma  \ref{Le.Uniform}, there exists some $j_0$ with  $j_0\in I_1\cup I_2$
	such that
	\begin{equation}\label{g121}
		\left|P_k\left(\theta_{j_0}-\frac{k-1}{2}\alpha\right)\right| \geq e^{kL-(\beta_n +\varepsilon)q_n}.
	\end{equation}

	First, assume $j_0\in I_2$.
	
	Set $I=[j_0-2sq_{n-n_0}+1,j_0+2sq_{n-n_0}-1]=[x_1,x_2]$.  By \eqref{g103},
	\begin{equation}\label{IIIcase1}
		|\phi(p)|  \leq \sum_{i=1,2} e^{\beta_n q_n+C\varepsilon q_n}|\phi(x_i^{\prime})|e^{-L|p-x_i| },
	\end{equation}
	where $x_1^{\prime}=x_1-1 $ and $x_2^{\prime}=x_2+1 $.
	
	For simplicity,	we are not going to  make the difference between $a\in\Z$ and $\tilde{a}\in\Z$ if  $|a-\tilde{a}|\leq 100\varepsilon q_n$. 
	Clearly, 
	\begin{equation}\label{g502}
		x_1^\prime\in \left[\ell q_n-\frac{1}{2}k_j, \ell q_n+\frac{1}{2}k_j\right], x_2^\prime\in\left[\ell q_n+\frac{3}{2}k_j, \ell q_n+\frac{5}{2}k_j\right].
	\end{equation}
	By Lemma \ref{le2}, one has that 
	\begin{equation}\label{g19}
		|\phi(x_2^\prime)|\leq r_{\ell+\eta}e^{C\varepsilon q_n-L|x_2-\ell q_n-k_j|}+r_{\ell+1}e^{C\varepsilon q_n-L|\ell q_n+ q_n-x_2|}.
	\end{equation}
	
	By Lemma \ref{le2} again, one has that 
	for any  $x_1^\prime\in [\ell q_n, \ell q_n+\frac{1}{2}k_j]$,  
	\begin{equation}\label{g20}
		|\phi(x_1^\prime)|\leq r_{\ell+\eta}e^{C\varepsilon q_n-L|\ell q_n+k_j-x_1|}+r_{\ell}e^{C\varepsilon q_n-L|x_1-\ell q_n|},
	\end{equation}
	and for any $x_1^\prime\in [\ell q_n-\frac{1}{2}k_j, \ell q_n]$,   
	\begin{equation}\label{g21}
		|\phi(x_1^\prime)|\leq r_{\ell-1+\eta}e^{C\varepsilon q_n-L|x_1-\ell q_n+q_n-k_j|}+r_{\ell}e^{C\varepsilon q_n-L|\ell q_n-x_1|}.
	\end{equation}
	By \eqref{IIIcase1}-\eqref{g21}, one has 
	\begin{eqnarray}
		|\phi(p)| &\leq& e^{C\varepsilon q_n-L k_j+\beta_n q_n} (r_{\ell+\eta}+r_{\ell})+e^{C\varepsilon q_n-L   q_n+\beta_n q_n}r_{\ell -1+\eta}
		\nonumber \\  &&+ e^{C\varepsilon q_n-L q_n+\beta_n q_n+Lk_j} r_{\ell+1}.\label{g27}
	\end{eqnarray}
	Therefore, we have
	\begin{equation}\label{g22}
		r_{\ell+\eta}\leq e ^{C\varepsilon q_n-L k_j+\beta_n q_n} (r_{\ell+\eta}+r_{\ell})+e^{C\varepsilon q_n-L   q_n+\beta_n q_n}r_{\ell -1+\eta}+ e^{C\varepsilon q_n-L q_n+\beta_n q_n+Lk_j} r_{\ell+1}.
	\end{equation}
	Since $L>\delta(\alpha,\theta)>\frac{\beta_nq_n}{k_j}$ (by \eqref{gnew2}), one has 
	\begin{equation}\label{gj161}
		e ^{C\varepsilon q_n-L k_j+\beta_n q_n} \leq \frac{1}{2}.
	\end{equation}
By \eqref{g22} and  \eqref{gj161}, we have that 
	\begin{equation}\label{g23}
		r_{\ell+\eta}\leq e ^{C\varepsilon q_n-L k_j+\beta_n q_n} r_{\ell}+e^{C\varepsilon q_n-L   q_n+\beta_n q_n}r_{\ell -1+\eta}+ e^{C\varepsilon q_n-L q_n+\beta_n q_n+Lk_j} r_{\ell+1}
	\end{equation}
	By \eqref{g500}, one has 
	\begin{equation}\label{g24}
		r_{\ell-1+\eta} \leq e^{C\varepsilon q_n} e^{Lq_n-Lk_j} r_{\ell}.
	\end{equation}
	By \eqref{g23} and \eqref{g24},
	one has that 
	\begin{equation}\label{g25}
		r_{\ell+\eta}\leq e^{C\varepsilon q_n+\beta_nq_n} (e ^{ -Lk_j} r_{\ell}+  e^{-L q_n+Lk_j} r_{\ell+1}).
	\end{equation}
	
	Thus in order to prove the  theorem, it suffices to exclude the case $j_0\in I_1$.
	
	Suppose  $j_0\in I_1$.  
	Following the proof of \eqref{g27}, we obtain  (move $-\ell q_n-k_j$ units in \eqref{g27})
	\begin{equation*}
		|\phi(0)|\leq \frac{1}{2}.
	\end{equation*}
	This  contradicts  $ \phi(0)=1$.
\end{proof}

\begin{theorem}\label{main4}
	Assume that $n=2j-1$  and $0<|\ell|\leq 20 \frac{b_{n+1}}{q_n} $. Then 
	\begin{equation}\label{g18}
		r_{\ell }\leq   e^{C\varepsilon q_n+3\beta_n q_n}\left( e^{-L q_n}r_{\ell +1}+ e^{ -Lq_n+Lk_j}r_{\ell -1+\eta}\right).
	\end{equation}
\end{theorem}
\begin{proof}
	Take any $p$ with $|p-\ell q_n|\leq  10\varepsilon q_n$  into consideration. 
	%
	%
	%
	
	Construct intervals
	\begin{equation*}
		I_1=\left[- \left\lfloor \frac{q_n}{2}\right \rfloor, q_n- \left\lfloor \frac{q_n}{2} \right\rfloor  -1\right], 
	\end{equation*}
	and
	\begin{equation*}
		I_2=\left[ \ell q_n- \left\lfloor \frac{q_n}{2} \right\rfloor,( \ell +1)q_n- \left\lfloor \frac{q_n}{2} \right\rfloor -1 \right].
	\end{equation*}
	
	Let $\theta_m=\theta+m\alpha$ for $m\in I_1\cup I_2$. The set $\{\theta_m\}_{m\in I_1\cup I_2}$
	consists of $2q_{n}$ elements. Let $k=2q_{n}-1$.
	In this case, similar to Claim 1, we can prove the claim (Claim 2):
	for any
	$  m\in I_1\cup I_2$, one has ${\rm Lag}_m\leq  (3\beta_n+\varepsilon) q_n  $.
	For convenience, we  include a proof in  Section \ref{Sclaim}.
	
	Applying   Lemma  \ref{Le.Uniform}, there exists some $j_0$ with  $j_0\in I_1\cup I_2$
	such that
	\begin{equation}\label{1g121}
		\left|P_k\left(\theta_{j_0}-\frac{k-1}{2}\alpha\right)\right| \geq e^{kL-(3\beta_n +\varepsilon)q_n}.
	\end{equation}

	First, assume $j_0\in I_2$.
	
	Set $I=[j_0-q_n+1,j_0+q_n-1]=[x_1,x_2]$.   By    \eqref{g103}, one has
	\begin{equation}\label{1IIIcase1}
		|\phi(p)|  \leq \sum_{i=1,2} e^{3\beta_n q_n+C\varepsilon q_n}|\phi(x_i^{\prime})|e^{-L|p-x_i| },
	\end{equation}
	where $x_1^{\prime}=x_1-1 $ and $x_2^{\prime}=x_2+1 $.
	
	Clearly, 
	\begin{equation}\label{g503}
		x_1^\prime\in \left[\ell q_n -\frac{3}{2} q_n, \ell q_n-\frac{1}{2} q_n\right], x_2^\prime\in\left[\ell q_n+\frac{1}{2}  q_n, \ell q_n+\frac{3}{2}q_n\right].
	\end{equation}
	
	By Lemma \ref{le2}, 
	\begin{itemize}
		\item  for $x_2^\prime\in [ \ell q_n+\frac{1}{2} q_n,\ell q_n+q_n]$,   
		\begin{equation}\label{1g20}
			|\phi(x_2^\prime)|\leq r_{\ell+\eta}e^{C\varepsilon q_n-L|x_2-\ell q_n-k_j|}+r_{\ell+1}e^{C\varepsilon q_n-L|\ell q_n+q_n-x_2|},
		\end{equation}
		\item 	for $x_2^\prime\in [\ell q_n+q_n, \ell q_n+q_n+k_j]$,   
		\begin{equation}\label{1g21}
			|\phi(x_2^\prime)|\leq r_{\ell+1+\eta}e^{C\varepsilon q_n-L|\ell q_n+q_n+k_j-x_2|}+r_{\ell+1}e^{C\varepsilon q_n-L|x_2-q_n-\ell q_n|}.
		\end{equation}
		\item for
		$x_2^\prime\in [ \ell q_n+q_n+k_j,\ell q_n+\frac{3}{2}q_n]$,
		\begin{equation}\label{11g21}
			|\phi(x_2^\prime)|\leq r_{\ell+1+\eta}e^{C\varepsilon q_n-L|x_2-\ell q_n-q_n-k_j|}+r_{\ell+2}e^{C\varepsilon q_n-L|\ell q_n+2q_n-x_2|}.
		\end{equation}
	\end{itemize}

	By Lemma \ref{le2} again, one has
	\begin{itemize}
		\item  for $x_1^\prime\in [ \ell q_n-q_n+ k_j,\ell q_n-\frac{1}{2} q_n]$,   
		\begin{equation}\label{11g20}
			|\phi(x_1^\prime)|\leq r_{\ell-1+\eta}e^{C\varepsilon q_n-L|x_1-\ell q_n+q_n-k_j|}+r_{\ell}e^{C\varepsilon q_n-L|\ell q_n-x_1|},
		\end{equation}
		\item 	for $x_1^\prime\in [\ell q_n-q_n, \ell q_n -q_n+k_j]$,   one has
		\begin{equation}\label{111g21}
			|\phi(x_1^\prime)|\leq r_{\ell-1+\eta}e^{C\varepsilon q_n-L|\ell q_n -q_n+k_j-x_1|}+r_{\ell-1}e^{C\varepsilon q_n-L|x_1-\ell q_n+q_n|}.
		\end{equation}
		\item for
		$x_1^\prime\in [\ell q_n-\frac{3}{2} q_n,\ell q_n-q_n]$   one has
		\begin{equation}\label{1111g21}
			|\phi(x_1^\prime)|\leq r_{\ell-2+\eta}e^{C\varepsilon q_n-L|x_1-\ell q_n+2q_n-k_j|}+r_{\ell-1}e^{C\varepsilon q_n-L|\ell q_n-q_n-x_1|}.
		\end{equation}
	\end{itemize}

	By \eqref{1IIIcase1}-\eqref{1111g21}, one has 
	\begin{eqnarray}
		|\phi(p)| &\leq& e^{C\varepsilon q_n-L q_n+3\beta_n  q_n}  r_{\ell}+e^{C\varepsilon q_n-L q_n+Lk_j+3\beta_n  q_n}  
		r_{\ell+\eta}
		\nonumber \\  &&+e^{C\varepsilon q_n-L   q_n-Lk_j+ 3\beta_n q_n}r_{\ell +1+\eta}+e^{C\varepsilon q_n-L   q_n+ 3\beta_n q_n}r_{\ell +1}	\nonumber\\
		&&+e^{C\varepsilon q_n-2L   q_n+ 3\beta_n q_n}r_{\ell +2} +e^{C\varepsilon q_n-L   q_n+Lk_j+ 3\beta_n q_n}r_{\ell -1+\eta} 	\nonumber \\
		&&+e^{C\varepsilon q_n-L   q_n+ 3\beta_n q_n}r_{\ell -1}+e^{C\varepsilon q_n-2L   q_n+Lk_j+ 3\beta_nq_n}r_{\ell -2+\eta}.
		\label{1g27}
	\end{eqnarray}
	Therefore, we have
	\begin{eqnarray}
		r_{\ell}&\leq& e^{C\varepsilon q_n-L q_n+3\beta_n  q_n}  r_{\ell}+e^{C\varepsilon q_n-L q_n+Lk_j+3\beta_n  q_n}  
		r_{\ell+\eta}
		\nonumber \\  &&+e^{C\varepsilon q_n-L   q_n-Lk_j+ 3\beta_n q_n}r_{\ell +1+\eta}+e^{C\varepsilon q_n-L   q_n+ 3\beta_n q_n}r_{\ell +1}	\nonumber\\
		&&+e^{C\varepsilon q_n-2L   q_n+ 3\beta_n q_n}r_{\ell +2} +e^{C\varepsilon q_n-L   q_n+Lk_j+ 3\beta_n q_n}r_{\ell -1+\eta} 	\nonumber \\
		&&+e^{C\varepsilon q_n-L   q_n+ 3\beta_n q_n}r_{\ell -1}+e^{C\varepsilon q_n-2L   q_n+Lk_j+ 3\beta_n q_n}r_{\ell -2+\eta}.
		\label{11g27}
	\end{eqnarray}
	
	By \eqref{g500}, one has 
	\begin{equation}\label{1g24}
		r_{\ell+\eta}\leq e^{C\varepsilon q_n} e^{L k_j} r_{\ell},
		r_{\ell+1+\eta} \leq e^{C\varepsilon q_n} e^{L k_j} r_{\ell+1},  r_{\ell+2} \leq e^{C\varepsilon q_n} e^{L q_n} r_{\ell+1},
	\end{equation}
	and 
	\begin{equation}\label{11g24}
		r_{\ell-1} \leq e^{C\varepsilon q_n} e^{Lk_j} r_{\ell-1+\eta},  r_{\ell-2+\eta} \leq e^{C\varepsilon q_n} e^{L q_n} r_{\ell-1+\eta},
	\end{equation}
	By 	 \eqref{11g27},  \eqref{1g24} and \eqref{11g24}, we have
	\begin{eqnarray}
		r_{\ell} &\leq& e^{C\varepsilon q_n-L q_n+2Lk_j+3\beta_n  q_n}  r_{\ell}+e^{C\varepsilon q_n-L   q_n+ 3\beta_n q_n}r_{\ell +1}
		\nonumber \\
		&&+e^{C\varepsilon q_n-L   q_n+Lk_j+ 3\beta_n q_n}r_{\ell -1+\eta}  .
		\label{101g27}
	\end{eqnarray}
By \eqref{gnew3}, one has 
	\begin{equation}\label{g00}
		e^{C\varepsilon q_n-L q_n+2L k_j+3\beta_n  q_n} \leq \frac{1}{2}.
	\end{equation}
	
	By \eqref{101g27} and \eqref{g00}, 
	one has
	\begin{equation}\label{1g25}
		r_{\ell }\leq    e^{C\varepsilon q_n-L   q_n+ 3\beta_n q_n}r_{\ell +1}+ e^{C\varepsilon q_n-L   q_n+Lk_j+ 3\beta_n q_n}r_{\ell -1+\eta}.
	\end{equation}
	
	Thus to prove the theorem, it suffices to exclude the case $j_0\in I_1$.
	Suppose  $j_0\in I_1$. 
	Following the proof of \eqref{1g25}, we have
	\begin{equation*}
		|\phi(0)|\leq \frac{1}{2}.
	\end{equation*}
	This  contradicts  $ \phi(0)=1$.
\end{proof}

\begin{theorem}\label{main5}
	Assume that $n=2j$ is  even and $0<|\ell|\leq 20 \frac{b_{n+1}}{q_n} $. Then 
	\begin{equation}\label{1g18}
		r_{\ell }\leq e^{C\varepsilon q_n-L q_n+\beta_n  q_n} \left (r_{\ell-1} +r_{\ell+1}\right).
	\end{equation}
\end{theorem}
\begin{proof}
	Take any $p$ with $|p-\ell q_n|\leq  10\varepsilon q_n$  into consideration.

	Construct intervals
	\begin{equation*}
		I_1=\left[- \left\lfloor \frac{q_n}{2}\right \rfloor, q_n- \left\lfloor \frac{q_n}{2} \right\rfloor  -1\right], 
	\end{equation*}
	and
	\begin{equation*}
		I_2=\left[ \ell q_n- \left\lfloor \frac{q_n}{2} \right\rfloor,( \ell +1)q_n- \left\lfloor \frac{q_n}{2} \right\rfloor -1 \right].
	\end{equation*}
	
	Let $\theta_m=\theta+m\alpha$ for $m\in I_1\cup I_2$. The set $\{\theta_m\}_{m\in I_1\cup I_2}$
	consists of $2q_{n}$ elements. Let $k=2q_{n}-1$.
	In this case, similar to Claims 1 and 2,  we can prove the claim (Claim 3):
	for any
	$  m\in I_1\cup I_2$, one has ${\rm Lag}_m\leq  (\beta_n +\varepsilon) q_n  $.
	For convenience, we  include a proof  in  Section \ref{Sclaim}.

	Applying   Lemma  \ref{Le.Uniform}, there exists some $j_0$ with  $j_0\in I_1\cup I_2$
	such that
	\begin{equation}\label{11g121}
		\left|P_k\left(\theta_{j_0}-\frac{k-1}{2}\alpha\right)\right| \geq e^{kL-(\beta_n +\varepsilon)q_n}.
	\end{equation}

	First assume $j_0\in I_2$.
	
	Set $I=[j_0-q_n+1,j_0+q_n-1]=[x_1,x_2]$.   By    \eqref{g103}, one has that
	\begin{equation}\label{11IIIcase1}
		|\phi(p)|  \leq \sum_{i=1,2} e^{\beta_n q_n+C\varepsilon q_n}|\phi(x_i^{\prime})|e^{-L|p-x_i| },
	\end{equation}
	where $x_1^{\prime}=x_1-1 $ and $x_2^{\prime}=x_2+1 $.
	
	Clearly, 
	\begin{equation}\label{9504}
		x_1^\prime\in \left[\ell q_n-\frac{3}{2} q_n, \ell q_n -\frac{1}{2} q_n\right], 
		x_2^\prime\in\left[\ell q_n+\frac{1}{2}  q_n,  \ell q_n+\frac{3}{2}q_n\right].
	\end{equation}
	By Lemma \ref{le2}, one has that
	\begin{itemize}
		\item  for $x_2^\prime\in [ \ell q_n+\frac{1}{2} q_n,\ell q_n+q_n]$,   
		\begin{equation}\label{111g20}
			|\phi(x_2^\prime)|\leq r_{\ell}e^{C\varepsilon q_n-L|x_2-\ell q_n|}+r_{\ell+1}e^{C\varepsilon q_n-L|\ell q_n+q_n-x_2|},
		\end{equation}
		\item for
		$x_2^\prime\in [ \ell q_n+q_n,\ell q_n+\frac{3}{2}q_n]$   one has
		\begin{equation}\label{112g21}
			|\phi(x_2^\prime)|\leq r_{\ell+1}e^{C\varepsilon q_n-L|x_2-\ell q_n-q_n|}+r_{\ell+2}e^{C\varepsilon q_n-L|\ell q_n+2q_n-x_2|}.
		\end{equation}
	\end{itemize}

	By Lemma \ref{le2} again, one has that
	\begin{itemize}
		\item  for $x_1^\prime\in [ \ell q_n-q_n,\ell q_n-\frac{1}{2}q_n]$,   
		\begin{equation}\label{112g20}
			|\phi(x_1^\prime)|\leq r_{\ell-1}e^{C\varepsilon q_n-L|x_1-\ell q_n+q_n|}+r_{\ell}e^{C\varepsilon q_n-L|\ell q_n-x_1|},
		\end{equation}
		\item 
	for	$x_1^\prime\in [ \ell q_n-\frac{3}{2}q_n,\ell q_n-q_n]$,
		\begin{equation}\label{11112g21}
			|\phi(x_1^\prime)|\leq r_{\ell-2}e^{C\varepsilon q_n-L|x_1-\ell q_n+2q_n|}+r_{\ell-1}e^{C\varepsilon q_n-L|\ell q_n-q_n-x_1|}.
		\end{equation}
	\end{itemize}

	By \eqref{11IIIcase1}-\eqref{11112g21}, one has that
	\begin{eqnarray}
		|\phi(p)| &\leq& e^{C\varepsilon q_n-L q_n+\beta_n  q_n}  (r_{\ell}+r_{\ell+1}+r_{\ell -1})
		\nonumber \\ 
		&&+e^{C\varepsilon q_n-2L   q_n+ \beta_n q_n} (r_{\ell+2}+r_{\ell -2}).
		\label{12g27}
	\end{eqnarray}
	Therefore, we have that
	\begin{eqnarray}
		r_{\ell}&\leq& e^{C\varepsilon q_n-L q_n+\beta_n  q_n}  (r_{\ell}+r_{\ell+1}+r_{\ell -1})
		\nonumber \\ 
		&&+e^{C\varepsilon q_n-2L   q_n+ \beta_n q_n} (r_{\ell+2}+r_{\ell -2}).
		\label{122g27}
	\end{eqnarray}
	Since $L>\frac{\beta}{2\eta}$, one has  that
	\begin{equation*}
		e^{C\varepsilon q_n-L q_n+\beta_n  q_n} \leq \frac{1}{2}.
	\end{equation*}
	By \eqref{122g27},  we have that
	\begin{equation}
		r_{\ell}\leq e^{C\varepsilon q_n-L q_n+\beta_n q_n}  (r_{\ell+1}+r_{\ell -1})
		+e^{C\varepsilon q_n-2L   q_n+ \beta_n q_n} (r_{\ell+2}+r_{\ell -2}).
		\label{1122g27}
	\end{equation}
	By \eqref{g500}, one has 
	\begin{equation}\label{12g24}
		r_{\ell+2} \leq e^{C\varepsilon q_n} e^{Lq_n} r_{\ell+1},  r_{\ell-2} \leq e^{C\varepsilon q_n} e^{L q_n} r_{\ell-1}.
	\end{equation}

	By \eqref{1122g27} and  \eqref{12g24}, 
	one has that
	\begin{equation}\label{12g25}
		r_{\ell }\leq e^{C\varepsilon q_n-L q_n+\beta_n  q_n}  (r_{\ell+1}+r_{\ell -1}) .
	\end{equation}
	
	Thus in order to prove the theorem, it suffices to exclude the case $j_0\in I_1$.
	
	Suppose  $j_0\in I_1$. 
	Following the proof of \eqref{12g27}, we have that
	\begin{equation*}
		|\phi(0)|\leq \frac{1}{2}.
	\end{equation*}
	This  contradicts  $ \phi(0)=1$.
\end{proof}
\begin{proof}[\bf Proof of Theorem \ref{main2}]
	Once we have Theorems \ref{main3}, \ref{main4}, \ref{main5} and Lemma \ref{le2} at hand,  Theorem \ref{main2} follows from standard iterations. 
	See ~\cite{jl18,liuetds20,l21} for example. 
\end{proof}

\section{Proof of Corollaries  and Claims 1-3}\label{Sclaim}

\begin{proof}[\bf Proof of Corollary \ref{main1}]
	It follows from 
	 Remark \ref{re}, \eqref{g6} and Theorem \ref{main2} with $\eta=10^{-2}$. 
\end{proof}

\begin{proof}[\bf Proof of Corollary \ref{cor1}]
	
	Let 	$a_0=0$, $a_1=\frac{2}{\mu}+10$, $q_{-1}=0$ and $q_{0}=1$.
	For $n\geq 1$, define $a_n$ and $q_n$  inductively  by 
	\begin{equation}\label{g1new}
	q_{n}=a_n q_{n-1}+q_{n-2}.
	\end{equation}
	and
		\begin{equation}\label{g2new}
	a_{n+1}=\left\lfloor e^{\mu q_{n}  } \right\rfloor.
	\end{equation}

	Let  $	\alpha=[a_1,a_2,\cdots,a_n,\cdots]$.
	By some simple facts of the continued fraction expansion, one has  that $\beta(\alpha)=\mu $ (see Section 2 for more details). Now Corollary \ref{cor1} follows from Corollary \ref{main1}.
\end{proof}
\begin{proof}[\bf Proof of Corollary \ref{cor2}]
It  follows from  Theorem \ref{main2}, Lemma \ref{le01} and  the  construction of $\alpha$ in Corollary \ref{cor1}. 
\end{proof}

\begin{proof}[\bf Proof of Claim 1]
	By the construction of $I_1$ and $I_2$ in Claim 1,  
	we have   the following estimates:
	\begin{itemize}
		\item  for any $j_1,j_2\in I_1 \cup I_2 $ with $j_1\neq j_2$,
		one has  that
		\begin{equation}\label{Appenoct1}
			j_1-j_2=l_1q_n+l_2,
		\end{equation}
		where $1\leq l_2<q_n$ and $|l_1|\leq 40\frac{b_{n+1}}{q_n}+4$. Therefore, by \eqref{gdc5},
		\begin{equation}\label{Appenoct2}
			||(j_1-j_2)\alpha||_{\R/\Z}\geq \frac{1}{4q_n}.
		\end{equation}
		\item 
		for any $j_1,j_2\in I_1 \cup I_2 $ with $j_1\neq j_2$,
		one has  
		\begin{equation}\label{Appenoct3}
			(j_1+j_2)=k_j+l_1q_n+l_2,
		\end{equation}
		where $0\leq l_2<q_n$ and $|l_1|\leq 40\frac{b_{n+1}}{q_n}+4$. Therefore by \eqref{gdc2},  for $l_2\geq 1$,
		\begin{equation}\label{Appenoct4}
			||2\theta+(j_1+j_2)\alpha||_{\R/\Z}\geq \frac{1}{4q_n},
		\end{equation}
		and  by \eqref{gdc1} and \eqref{g0},
		for $l_2= 0$
		\begin{equation}\label{Appenoct5}
			||2\theta+(j_1+j_2)\alpha||_{\R/\Z} \geq  \frac{\eta}{10 q_{n+1}}\geq e^{-(\beta_n+\varepsilon)q_n}.
		\end{equation}
	\end{itemize}
	
	Moreover, for $j_1\in I_1\cup I_2$, there is at most one $j_2\in I_1\cup I_2$  such that the $l_2$ in 
	\eqref{Appenoct3} is 0. 
	By the standard arguments (e.g.  Appendices in ~\cite{jl18,liuetds20}), we have that  for any $m\in I_1\cup l_2$,
	\begin{equation*}
		{\rm Lag}_m\leq \beta_n q_n+ \varepsilon q_n.
	\end{equation*}
	
\end{proof}

\begin{proof}[\bf Proof of Claim 2]
	By the construction of $I_1$ and $I_2$ in Claim 2,  and Proposition \ref{prop},
	we have   the following estimates:
	\begin{itemize}
		\item  for any $j_1,j_2\in I_1 \cup I_2 $ with $j_1\neq j_2$,
		one has  that
		\begin{equation}\label{1Appenoct1}
			j_1-j_2=l_1q_n+l_2,
		\end{equation}
		where $0\leq l_2<q_n$ and $|l_1|\leq 40\frac{b_{n+1}}{q_n}+4$. Therefore, for $l_2\neq 0$
		\begin{equation}\label{1Appenoct2}
			||(j_1-j_2)\alpha||_{\R/\Z}\geq \frac{1}{4q_n},
		\end{equation}
		and 
		for $l_2= 0$,   
		\begin{equation}\label{11Appenoct2}
			||(j_1-j_2)\alpha||_{\R/\Z}\geq  \frac{1}{2 q_{n+1}}\geq e^{-(\beta_n+\varepsilon)q_n}.
		\end{equation}
		\item  for any $j_1,j_2\in I_1 \cup I_2 $ with $j_1\neq j_2$,
		one has  
		\begin{equation}\label{1Appenoct3}
			(j_1+j_2)=k_j+l_1q_n+l_2,
		\end{equation}
		where $0\leq l_2<q_n$ and $|l_1|\leq 40\frac{b_{n+1}}{q_n}+4$. Therefore, for $l_2\neq 0$
		\begin{equation}\label{1Appenoct4}
			||2\theta+(j_1+j_2)\alpha||_{\R/\Z}\geq \frac{1}{4q_n},
		\end{equation}
		and 
		for $l_2= 0$
		\begin{equation}\label{11Appenoct5}
			||2\theta+(j_1+j_2)\alpha||_{\R/\Z} \geq  \frac{\eta}{10 q_{n+1}}\geq e^{-(\beta_n+\varepsilon)q_n}.
		\end{equation}
	\end{itemize}
	
	Moreover, for any $j_1\in I_1\cup I_2$, there is at most one $j_2\in I_1\cup I_2$ with  $j_2\neq j_1$ such that the $l_2$ in 
	\eqref{1Appenoct1} is 0.  For any $j_1\in I_1\cup I_2$, there is at most two $j_2\in I_1\cup I_2$ with  $j_2\neq j_1$ such that the $l_2$ in 
	\eqref{1Appenoct3} is 0. 
	Following the discussion in Claim 1,  we have that  for any $m\in I_1\cup l_2$,
	\begin{equation*}
		{\rm Lag}_m\leq 3\beta_n q_n+ \varepsilon q_n.
	\end{equation*}
	
\end{proof}
\begin{proof}[\bf Proof of Claim 3]
	By the construction of $I_1$ and $I_2$ in Claim 3,  and Proposition \ref{prop},
	we have   the following estimates:
	\begin{itemize}
		\item  for any $j_1,j_2\in I_1 \cup I_2 $ with $j_1\neq j_2$,
		one has  that 
		\begin{equation}\label{21Appenoct1}
			j_1-j_2=l_1q_n+l_2,
		\end{equation}
		where $0\leq l_2<q_n$ and $|l_1|\leq 40\frac{b_{n+1}}{q_n}+4$. Therefore, for $l_2\neq 0$
		\begin{equation}\label{21Appenoct2}
			||(j_1-j_2)\alpha||_{\R/\Z}\geq \frac{1}{4q_n},
		\end{equation}
		and 
		for $l_2= 0$
		\begin{equation}\label{211Appenoct2}
			||(j_1-j_2)\alpha||_{\R/\Z}\geq  \frac{1}{2 q_{n+1}}\geq e^{-(\beta_n+\varepsilon)q_n}.
		\end{equation}
		\item  for any $j_1,j_2\in I_1 \cup I_2 $,
		\begin{equation}\label{211Appenoct5}
			||2\theta+(j_1+j_2)\alpha||_{\R/\Z} \geq  \frac{\eta}{20 q_{n}}.
		\end{equation}
	\end{itemize}
	
	Moreover, for $j_1\in I_1\cup I_2$, there is at most one $j_2\in I_1\cup I_2$ with  $j_2\neq j_1$ such that the $l_2$ in 
	\eqref{21Appenoct1} is 0.  
	Following the discussion in Claim 1,   we have that  for any $m\in I_1\cup l_2$,
	\begin{equation*}
		{\rm Lag}_m\leq \beta_n q_n+ \varepsilon q_n.
	\end{equation*}
	
\end{proof}
\section*{Acknowledgments}
W. L. was  supported by   NSF  DMS-2000345 and DMS-2052572.

\bibliographystyle{abbrv} 
\bibliography{Exa}

\end{document}